\documentclass[conference, twocolumns, usletter]{IEEEtran}
\usepackage{amsmath,amssymb,amsthm,mathrsfs,amsfonts,dsfont,stmaryrd,wasysym,marvosym}
\usepackage{mathtools}
\usepackage{color}
\usepackage{graphicx}  
\usepackage{psfrag, xcolor}
\usepackage{enumitem}
\usepackage{subfigure,cite,epsfig,dblfloatfix,url}

\allowdisplaybreaks[4]
\newtheorem{theorem}{Theorem}
\newtheorem{lemma}{Lemma}  
  
\newtheorem{proposition}{Proposition}   
\newtheorem{definition}{Definition} 
\newtheorem{remark}{Remark} 
\newcommand{\dv}{\mathbf} 
\newcommand{\mc}{\mathcal} 
\newcommand{\mkv}{-\!\!\!\!\minuso\!\!\!\!-}

\newcommand{\bqed}{\tag*{$\blacksquare$}}
\newcommand*{\qedblack}{\hfill\ensuremath{\blacksquare}}

\newcommand{\squeezeup}{\vspace{-1em}}

\newcommand*\xbar[1]{%
    \hbox{%
		 \vbox{%
		 \hrule height 0.5pt 
		 \kern0.5ex
		 \hbox{%
		 \kern-0.1em
		\ensuremath{#1}%
		\kern-0.1em
		}%
		}%
		}%
		} 
		
\graphicspath{{figs/}}

\begin{document}
\fontencoding{OT1}\fontsize{9.4}{11}\selectfont

\title{Optimal Rate-Exponent Region for a Class of Hypothesis Testing Against Conditional Independence Problems}
\author{Abdellatif Zaidi$^{\dagger}$ $^{\ddagger}$  \qquad \qquad Inaki Estella Aguerri$^{\dagger}$ \vspace{0.3cm}\\
 $^{\dagger}$ Paris Research Center, Huawei Technologies, Boulogne-Billancourt, 92100, France\\
$^{\ddagger}$ Universit\'e Paris-Est, Champs-sur-Marne, 77454, France\\
\{abdellatif.zaidi@u-pem.fr, inaki.estella@gmail.com\}
\vspace{-5mm}
}

\maketitle

\begin{abstract}
We study a class of distributed hypothesis testing against conditional independence problems. Under the criterion that stipulates minimization of the Type II error rate subject to a (constant) upper bound $\epsilon$ on the Type I error rate, we  characterize the set of encoding rates and exponent for both discrete memoryless and memoryless vector Gaussian settings.
\end{abstract}

\section{Introduction}

Consider the multiterminal detection system shown in Figure~\ref{fig-distributed-hypothesis-testing}. In this problem, a memoryless vector source $(X,Y_0,Y_1,\hdots,Y_K)$, $K \geq 1$, has joint distribution that depends on two hypotheses, a null hypothesis $H_0$ and an alternate hypothesis $H_1$. A detector that observes directly the pair $(X,Y_0)$ but only receives summary information of the observations $(Y_1,\hdots,Y_K)$ seeks to determine which of the two hypotheses is true. Specifically, Encoder $k$, $1 \leq k \leq K$, which observes an i.i.d. string $Y^n_k$, sends a message $M_k$ to the detector a finite rate of $R_k$ bits per observation over a noise-free channel; and the detector makes its decision between the two hypotheses on the basis of the received messages $(M_1,\hdots,M_K)$ as well as the available pair $(X^n,Y^n_0)$. In doing so, the detector can make two types of error: Type I error (guessing $H_1$ while $H_0$ is true) and Type II error (guessing $H_0$ while $H_1$ is true). The type II error probability decreases exponentially fast with the size $n$ of the i.i.d. strings, say with an exponent $E$; and, classically, one is interested is characterizing the set of achievable rate-exponent tuples $(R_1,\hdots,R_K,E)$ in the regime in which the probability of the Type I error is kept below a prescribed small value $\epsilon$. This problem, which was first introduced by Berger~\cite{B79}, and then studied further in~\cite{AC86,H87,SP92}, arises naturally in many applications (for recent developments on this topic, the reader may refer to~\cite{TC08,ZL15,SWT18,EWZ18,EZW18,LSCT17,SGC18} and references therein). Its theoretical understanding, however, is far from complete, even from seemingly simple instances of it.

\begin{figure}[h!]
\centering
\includegraphics[width=0.7\linewidth]{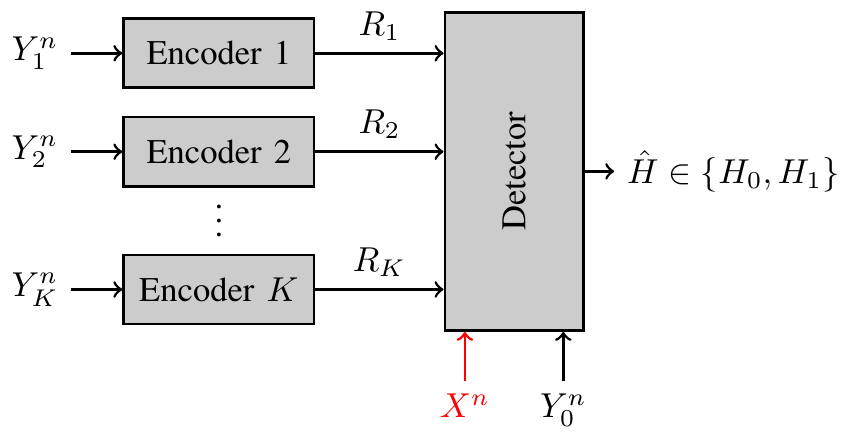}
\caption{Distributed hypothesis testing against conditional independence.} 
\squeezeup
\label{fig-distributed-hypothesis-testing}
\end{figure}

One important such instances was studied by Rahman and Wagner in~\cite{RW12}. In~\cite{RW12}, the two hypotheses are such that  
 $X$ and $(Y_1,\hdots,Y_K)$ are correlated conditionally given $Y_0$ under the null hypothesis $H_0$; and they are independent conditionally given $Y_0$ under the alternate hypothesis $H_1$, i.e.,~\footnote{In fact, the model of~\cite{RW12} also involves a random variable $Y_{K+1}$, which is chosen here to be deterministic as it is not relevant for the analysis and discussion that will follow in this paper (see Remark $2$).}
\begin{subequations}
\begin{align}
\label{distribution-under-null-hypothesis-special-case-ht-against-conditional-independence-DM-Rahman-Wagner}
H_0  &: P_{X,Y_0,Y_1\hdots,Y_K} = P_{Y_0} P_{X,Y_1,\hdots,Y_K|Y_0} \\
H_1  &: Q_{X,Y_0,Y_1\hdots,Y_K} = P_{Y_0} P_{X|Y_0} P_{Y_1,\hdots,Y_K|Y_0}. 
\label{distribution-under-alternate-hypothesis-special-case-ht-against-conditional-independence-DM-Rahman-Wagner}
\end{align}
\label{distributions-under-null-and-alternate-hypotheses-special-case-ht-against-conditional-independence-DM-Rahman-Wagner}
\end{subequations}
Note that $(Y_0,Y_1,\hdots,Y_K)$ and $(Y_0,X)$ have the same distributions under both hypotheses; and the multiterminal problem~\eqref{distributions-under-null-and-alternate-hypotheses-special-case-ht-against-conditional-independence-DM-Rahman-Wagner} is a generalization of the single-encoder test against independence studied by Ahlswede and Csiszar in~\cite{AC86}. For the problem~\eqref{distributions-under-null-and-alternate-hypotheses-special-case-ht-against-conditional-independence-DM-Rahman-Wagner} Rahman and Wagner provided inner and outer bounds on the rate-exponent region which do \textit{not} match in general (see~\cite[Theorem 1]{RW12} for the inner bound and~\cite[Theorem 2]{RW12} for the outer bound). The inner bound of~\cite[Theorem 1]{RW12} is based on a scheme, named Quantize-Bin-Test scheme therein, that is similar to the Berger-Tung distributed source coding scheme~\cite{B77,T78}.

In this paper, we study a class of the hypothesis testing problem~\eqref{distributions-under-null-and-alternate-hypotheses-special-case-ht-against-conditional-independence-DM-Rahman-Wagner} obtained by restricting the joint distribution of the variables under the null hypothesis to satisfy the Markov chain
\begin{equation}
Y_{\mc S} \mkv  (X, Y_0)  \mkv Y_{{\mc S}^c} \quad \forall \:\: \mc S \subseteq \mc K := \{1,\ldots,K\}
\label{markov-chain}
\end{equation}
i.e., the encoders' observations $\{Y_k\}_{k \in \mc K}$ are independent conditionally given $(X,Y_0)$ under $H_0$. We investigate both discrete memoryless (DM) and memoryless vector Gaussian models. For the DM setting, we provide a converse proof and show that it is achieved using the Quantize-Bin-Test scheme of~\cite[Theorem 1]{RW12}. Our converse proof is inspired by that of the rate-distortion region of the Chief-Executive Officer (CEO) problem under logarithmic loss of Courtade and Weissman~\cite[Theorem 10]{CW14}. We note that, prior to this work, for general distributions under the null hypothesis (i.e., without the Markov chain~\eqref{markov-chain} under this hypothesis) the optimality of the Quantize-Bin-Test scheme of~\cite{RW12} for the problem of testing against conditional independence was known only for the special case of a single encoder, i.e., $K=1$, (see~\cite[Theorem 3]{RW12}), a result which can also be recovered from our result in this paper. 

For the vector Gaussian setting too we provide a full characterization of the rate-exponent region. For the proof of the converse of this result, we obtain an outer bound by evaluating our outer bound the DM model by means of a technique that relies on the de Bruijn identity and the properties of Fisher information. In doing so, we show that the Quantize-Bin-Scheme of~\cite[Theorem 1]{RW12} with Gaussian test channels and time-sharing is optimal, thus providing what appears to be the first optimality result for the Gaussian hypothesis testing against conditional independence problem in the vector case even for the single-encoder case, i.e., $K=1$. 
 
\subsection{Notation}

Throughout this paper, we use the following notation. Upper case letters are used to denote random variables, e.g., $X$; lower case letters are used to denote realizations of random variables, e.g., $x$; and calligraphic letters denote sets, e.g., $\mc X$.  The cardinality of a set $\mc X$ is denoted by $|\mc X|$. The closure of a set $\mc A$ is denoted by $\xbar{A}$. The length-$n$ sequence $(X_1,\ldots,X_n)$ is denoted as  $X^n$; and, for integers $j$ and $k$ such that $1 \leq k \leq j \leq n$, the sub-sequence $(X_k,X_{k+1},\ldots, X_j)$ is denoted as  $X_{k}^j$. Probability mass functions (pmfs) are denoted by $P_X(x)=\mathrm{Pr}\{X=x\}$; and, sometimes, for short, as $p(x)$. We use $\mc P(\mc X)$ to denote the set of discrete probability distributions on $\mc X$.  Boldface upper case letters denote vectors or matrices, e.g., $\dv X$, where context should make the distinction clear. For an integer $K \geq 1$, we denote the set of integers smaller or equal $K$ as $\mc K = \{ k \in \mathbb{N} \: : \: 1 \leq k \leq K\}$. For a set of integers $\mc S \subseteq \mc K$, the complementary set of $\mc S$ is denoted by $\mc S^c$, i.e., $\mc S^c = \{k \in \mathbb{N} \: : \: k \in \mc K \setminus \mc S\}$. Sometimes, for convenience we will need to define $\bar{\mc S}$ as $\bar{\mc S} = \{0\} \cup \mc S^c$. For a set of integers $\mc S \subseteq \mc K$; the notation $X_{\mc S}$ designates the set of random variables $\{X_k\}$ with indices in the set $\mc S$, i.e., $X_{\mc S}=\{X_k\}_{k \in \mc S}$. We denote the covariance of a zero mean, complex-valued, vector $\dv X$ by $\mathbf{\Sigma}_{\mathbf{x}} =\mathbb{E}[\mathbf{XX}^{\dag}]$, where $(\cdot)^{\dag}$ indicates conjugate transpose. Similarly, we denote the cross-correlation of two zero-mean vectors $\dv X$  and $\dv Y$ as $\mathbf{\Sigma}_{\mathbf{x},\mathbf{y}} = \mathbb{E}[\mathbf{XY}^{\dag}]$, and the conditional correlation matrix of $\mathbf{X}$ given $\mathbf{Y}$ as $\mathbf{\Sigma}_{\mathbf{x}|\mathbf{y}} = \mathbb{E}\big[\big(\dv X - \mathbb{E}[\dv X|\dv Y]\big)\big(\dv X - \mathbb{E}[\dv X|\dv Y]\big)^{\dag}\big]$ i.e., $\mathbf{\Sigma}_{\mathbf{x}|\mathbf{y}} = \mathbf{\Sigma}_{\mathbf{x}}-\mathbf{\Sigma}_{\mathbf{x},\mathbf{y}}\mathbf{\Sigma}_{\mathbf{y}}^{-1}\mathbf{\Sigma}_{\mathbf{y},\mathbf{x}}$. For matrices $\dv A$ and $\dv B$, the notation $\mathrm{diag}(\dv A, \dv B)$ denotes the block diagonal matrix whose diagonal elements are the matrices $\dv A$ and $\dv B$ and its off-diagonal elements are the all zero matrices. Also, for a set of integers $\mc J \subset \mathbb{N}$ and a family of matrices $\{\dv A_i\}_{i \in \mc J}$ of the same size, the notation $\dv A_{\mc J}$ is used to denote the (super) matrix obtained by concatenating vertically the matrices $\{\dv A_i\}_{i \in \mc J}$, where the indices are sorted in the ascending order, e.g, $\dv A_{\{0,2\}}=[\dv A^{\dag}_0, \dv A^{\dag}_2]^{\dag}$.

\section{Problem Formulation}~\label{secII}

Consider a $(K+2)$-dimensional memoryless source $(X,Y_0,Y_1,\ldots,Y_K)$ with finite alphabet $\mc X \times \mc Y_0 \times \mc Y_1 \times \ldots \times \mc Y_K$. The joint probability mass function (pmf) of $(X,Y_0,Y_1,\ldots,Y_K)$ is assumed to be determined by a hypothesis $H$  that takes one of two values, a null hypothesis $H_0$ and an alternate hypothesis $H_1$. Specifically,  $X$ and $(Y_0,Y_1,\hdots,Y_K)$ are correlated under the null hypothesis $H_0$, with their joint distribution assumed to satisfy the Markov chain
\begin{equation}
Y_{\mc S} \mkv  (X, Y_0)  \mkv Y_{{\mc S}^c} \quad \forall \:\: \mc S \subseteq \mc K := \{1,\ldots,K\}
\label{eq:MKChain_pmf}
\end{equation}
under this hypothesis; and $X$ and $(Y_1,\hdots,Y_K)$ are independent conditionally given $Y_0$ under the alternate hypothesis $H_1$, i.e.,
\begin{subequations}
\begin{align}
\label{distribution-under-null-hypothesis-ht-against-conditional-independence-DM}
H_0  &: P_{X,Y_0,Y_1\hdots,Y_K} = P_{X,Y_0} \prod_{i=1}^K P_{Y_k|X,Y_0} \\
H_1  &: Q_{X,Y_0,Y_1\hdots,Y_K} = P_{Y_0} P_{X|Y_0} P_{Y_1,\hdots,Y_K|Y_0}. 
\label{distribution-under-alternate-hypothesis-ht-against-conditional-independence-DM}
\end{align}
\end{subequations}
Let now $\{(X_i,Y_{0,i},Y_{1,i},\ldots,Y_{K,i})\}^n_{i=1}$ be a sequence of $n$ independent copies of $(X,Y_0,Y_1,\ldots,Y_K)$; and consider the detection system shown in Figure~\ref{fig-distributed-hypothesis-testing}. Here, there are $K$ sensors and one detector. Sensor $k \in \mc K$ observes the memoryless source component $Y^n_k$ and sends a message $M_k = \breve{\phi}^{(n)}_k(Y^n_k)$ to the detector, where the mapping
\begin{equation}
\breve{\phi}^{(n)}_k \: : \:  \mc Y^n_k \rightarrow \{1,\ldots,M^{(n)}_k\}
\end{equation}
designates the encoding operation at this sensor. The detector observes the pair $(X^n,Y^n_0)$ and uses them, as well as the messages $\{M_1,\hdots,M_K\}$ gotten from from the sensors, to make a decision between the two hypotheses, based on a decision rule
\begin{equation}
\breve{\psi}^{(n)} \: : \{1,\ldots,M^{(n)}_1\}\times \ldots \times \{1,\ldots,M^{(n)}_K\}\times \mc X^n \times \mc Y_0^n \rightarrow \{H_0,H_1\}.
\label{decision-rule-hypothesis-testing}
\end{equation}
The mapping~\eqref{decision-rule-hypothesis-testing} is such that $\breve{\psi}^{(n)}(m_1,\hdots,m_K,x^n,y^n_0)=H_0$ if $(m_1,\hdots,m_K,x^n,y^n_0) \in \mc A_n$ and $H_1$ otherwise, with
\begin{equation*}
\mc A_n \subseteq \prod_{k=1}^n \{1,\ldots,M^{(n)}_k\} \times \mc X^n \times \mc Y_0^n
\end{equation*}
designating the acceptance region for $H_0$. The encoders $\{\breve{\phi}^{(n)}_k\}_{k=1}^K$ and the detector $\breve{\psi}^{(n)}$ are such that the Type I error probability does not exceed a prescribed level $\epsilon \in [0,1]$, i.e.,
\begin{equation}
P_{\breve{\phi}^{(n)}_1(Y^n_1),\hdots,\breve{\phi}^{(n)}_K(Y^n_K),X^n,Y^n_0} (\mc A^c_n) \leq \epsilon
\end{equation}
and the Type II error probability does not exceed $\beta$, i.e.,
\begin{equation}
Q_{\breve{\phi}^{(n)}_1(Y^n_1),\hdots,\breve{\phi}^{(n)}_K(Y^n_K),X^n,Y^n_0} (\mc A_n) \leq \beta.
\end{equation}
\begin{definition}
A rate-exponent tuple $(R_1,\hdots,R_K,E)$ is achievable for a fixed $\epsilon \in [0,1]$ if for any positive $\delta$ and sufficiently large $n$ there exist encoders $\{\breve{\phi}^{(n)}_k\}_{k=1}^K$ and a detector $\breve{\psi}^{(n)}$ such that
\begin{subequations}
\begin{align}
\frac{1}{n} \log M^{(n)}_k &\leq R_k + \delta \:\: \text{for all}\:\: k \in \mc K, \:\: \text{and} \\
 -\frac{1}{n} \log \beta &\geq E - \delta. 
\end{align}
\end{subequations}
The rate-exponent region $\mc R_{\text{HT}}$ is defined as
\begin{equation}
\mc R_{\text{HT}} := \bigcap_{\epsilon > 0} \mc R_{\text{HT},\epsilon},
\label{definition-rate-exponent-region}
\end{equation}
where $\mc R_{\text{HT},\epsilon}$ is the set of all achievable rate-exponent vectors for a fixed $\epsilon \in [0,1]$. 
\qed
\end{definition}

\section{Discrete Memoryless Case}~\label{secIII}

 We start with an entropy characterization of the rate-exponent region $\mc R_{\text{HT}}$ as defined by~\eqref{definition-rate-exponent-region}. Let
\begin{equation}
\mc R^{\star} = \bigcup_{n} \bigcup_{\{\breve{\phi}^{(n)}_k\}_{k \in \mc K}}  \mc R^{\star}\Big(n,\{\breve{\phi}^{(n)}_k\}_{k \in \mc K}\Big)
\end{equation}
where
\begin{subequations}
\begin{align}
& \mc R^{\star}\Big(n,  \{\breve{\phi}^{(n)}_k\}_{k \in \mc K}\Big) =  \Big\{(R_1,\hdots,R_K,E) \:\: \text{s.t.} \nonumber\\
&\hspace{1.5cm} R_k \geq \frac{1}{n} \log|\breve{\phi}^{(n)}_k(Y^n_k)| \:\: \text{for all}\:\: k \in \mc K, \:\: \text{and} \\
&\hspace{1.5cm} E \leq \frac{1}{n} I(\{\breve{\phi}^{(n)}_k(Y^n_k)\}_{k \in \mc K};X^n|Y^n_0) \Big\}.
\end{align}
\label{entropy-characterization-rate-exponent-region-distributed-ht-against-conditional-independence}
\end{subequations}

\noindent We have the following proposition the proof of which is essentially similar to that of \cite[Theorem 5]{AC86}; and, hence, is omitted.

\begin{proposition}~\label{proposition-entropy-characterization-rate-exponent-region-distributed-ht-against-conditional-independence}
$\mc R_{\text{HT}} = \xbar{\mc R^{\star}}$.
\end{proposition}

\noindent We now have the following theorem which provides a single-letter characterization of the rate-exponent region $\mc R_{\text{HT}}$.

\begin{theorem}~\label{theorem-rate-exponent-region-hypothesis-testing-DM-case}
The rate-exponent region $\mc R_{\text{HT}}$ is given by the union of all  non-negative tuples $(R_1,\ldots,R_K,E)$ that satisfy, for all subsets $\mc S \subseteq \mc K$, 
\begin{equation*}
E \leq I(U_{\mc S^c};X|Y_0,Q) + \sum_{k \in \mc S} \big(R_k - I(Y_k;U_k|X,Y_0,Q)\big)
\end{equation*}
for some auxiliary random variables $(U_1,\ldots,U_K,Q)$ with distribution $P_{U_{\mc K},Q}(u_{\mc K},q)$ such that 
\begin{align}
& P_{X, Y_0, Y_{\mc K}, U_{\mc K}, Q}(x, y_0, y_{\mc K},u_{\mc K},q) = P_Q(q)  P_{X, Y_0}( x, y_0) \nonumber\\
& \hspace{0.1cm} = \prod_{k=1}^K P_{Y_k|X, Y_0}(y_k|x,y_0) \: \prod_{k=1}^{K} P_{U_k|Y_k,Q}(u_k|y_k,q).
\end{align}
\end{theorem}

\begin{proof}
The proof of Theorem~\ref{theorem-rate-exponent-region-hypothesis-testing-DM-case} is given in Section~\ref{secV_subsecA}. 
\end{proof}

\begin{remark}\label{remark-optimality-QBT-DM-case}
As we mentioned in the introduction section, Rahman and Wagner~\cite{RW12} study the hypothesis testing problem of Figure~\ref{fig-distributed-hypothesis-testing} in the case in which $X$ is replaced by a two-source $(Y_{K+1},X)$ such that, like in our setup (which corresponds to $Y_{K+1}$ deterministic), $Y_0$ induces conditional independence between $(Y_1,\hdots,Y_K,Y_{K+1})$ and $X$ under the alternate hypothesis $H_1$. Under the null hypothesis $H_0$, however, the model studied by Rahman and Wagner in~\cite{RW12} assumes a more general distribution than ours in which $(Y_1,\hdots,Y_K,Y_{K+1})$ are arbitrarily correlated among them and with the pair $(X,Y_0)$. More precisely, the joint distributions of $(X,Y_1,\hdots,Y_K,Y_{K+1})$ under the null and alternate hypotheses as considered in~\cite{RW12} are
\begin{subequations}
\begin{align}
\label{distribution-under-null-hypothesis-ht-against-conditional-independence-DM-Rahman-Wagner}
H_0  &: \tilde{P}_{X,Y_0,Y_1\hdots,Y_K,Y_{K+1}} = P_{Y_0} P_{X,Y_1,\hdots,Y_K,Y_{K+1}|Y_0} \\
H_1  &: \tilde{Q}_{X,Y_0,Y_1\hdots,Y_K,Y_{K+1}} = P_{Y_0} P_{X|Y_0} P_{Y_1,\hdots,Y_K,Y_{K+1}|Y_0}. 
\label{distribution-under-alternate-hypothesis-ht-against-conditional-independence-DM-Rahman-Wagner}
\end{align}
\label{distributions-under-null-and-alternate-hypotheses-ht-against-conditional-independence-DM-Rahman-Wagner}
\end{subequations}
\noindent For this model, they provide inner and outer bounds on the rate-exponent region which do not mach in general (see~\cite[Theorem 1]{RW12} for the inner bound and~\cite[Theorem 2]{RW12} for the outer bound). The inner bound of~\cite[Theorem 1]{RW12} is based on a scheme, named Quantize-Bin-Test scheme therein, that is similar to the Berger-Tung distributed source coding scheme~\cite{B77,T78}; and whose achievable rate-exponent region can be shown through submodularity arguments to be equivalent to the region stated in Theorem~\ref{theorem-rate-exponent-region-hypothesis-testing-DM-case} (with $Y_{K+1}$ set to be deterministic). The result of Theorem~\ref{theorem-rate-exponent-region-hypothesis-testing-DM-case} then shows that if the joint distribution of the variables under the null hypothesis is restricted to satisfy~\eqref{distribution-under-null-hypothesis-ht-against-conditional-independence-DM}, i.e., the encoders' observations $\{Y_k\}_{k \in \mc K}$ are independent conditionally given $(X,Y_0)$, then the Quantize-Bin-Test scheme of~\cite[Theorem 1]{RW12} is optimal. We note that, prior to this work, for general distributions under the null hypothesis (i.e., without the Markov chain~\eqref{eq:MKChain_pmf} under this hypothesis) the optimality of the Quantize-Bin-Test scheme of~\cite{RW12} for the problem of testing against conditional independence was known only for the special case of a single encoder, i.e., $K=1$, (see~\cite[Theorem 3]{RW12}), a result which can also be recovered from Theorem~\ref{theorem-rate-exponent-region-hypothesis-testing-DM-case}. 
\end{remark}

\section{Memoryless Vector Gaussian Case}~\label{secIV}

We now turn to a continuous example of the hypothesis testing problem studied in this paper. Here,  $(\dv X,\dv Y_0,\dv Y_1,\hdots,\dv Y_K)$ is a zero-mean Gaussian random vector such that 
\begin{equation}~\label{mimo-gaussian-ht-model-2}
\dv Y_0 = \dv H_0 \dv X  + \dv N_0
\end{equation}
 where $\dv H_0 \in \mathds{C}^{n_0\times n_x}$, $\dv X \in \mathds{C}^{n_x}$ and $\dv N_0 \in \mathds{C}^{n_0}$ are independent Gaussian vectors with zero-mean and covariance matrices $\dv\Sigma_{\dv x} \succeq \dv 0$ and $\dv\Sigma_0 \succeq \dv 0$, respectively. The vectors $(\dv Y_1, \hdots,\dv Y_K)$ and $\dv X$ are correlated under the null hypothesis $H_0$ and are independent under the alternate hypothesis $H_1$, with
\begin{subequations}
\begin{align}
\label{distributions-under-null-hypothesis-ht-against-conditional-independence-vector-Gaussian}
H_0 \: &: \dv Y_k = \dv H_k \dv X + \dv N_k, \quad\text{for all}\:\: k \in \mc K\\
H_1 \: &:  (\dv Y_1,\hdots,\dv Y_K) \:\: \text{independent from}\:\: \dv X \:\: \text{conditionally given}\:\: \dv Y_0.
\label{distributions-under-alternate-hypothesis-ht-against-conditional-independence-vector-Gaussian}
\end{align}
\label{distributions-under-null-and-alternate-hypotheses-ht-against-conditional-independence-vector-Gaussian}
\end{subequations}
The noise vectors $(\dv N_1,\hdots,\dv N_K)$ are jointly Gaussian with zero mean and covariance matrix $\dv\Sigma_{\dv n_{\mc K}}  \succeq \dv 0$. They are assumed to be independent from $\dv X$ but correlated among them and with $\dv N_0$, with for every $\mc S \subseteq \mc K$, 
\begin{equation}  
\dv N_{\mc S} \mkv  \dv N_0  \mkv \dv N_{\mc S^c}. 
\label{markov-chain-assumption-gaussian-hypothesis-testing-model}
\end{equation}
Let $\dv\Sigma_k$ denote the covariance matrix of noise $\dv N_k$, $k \in \mc K$. Also, let $\mc R_{\text{VG-HT}}$ denote the rate-exponent region of this vector Gaussian hypothesis testing against conditional independence problem. 

\noindent For convenience, we now introduce the following notation which will be instrumental in what follows. Let, for every set $\mc S \subseteq \mc K$, the set $\bar{\mc S} = \{0\} \cup \mc S^c$. Also, for $\mc S \subseteq \mc K$ and given matrices $\{\dv\Omega_k\}_{k=1}^K$ such that $\dv 0 \preceq \dv\Omega_k \preceq \dv\Sigma_k^{-1}$, let $\boldsymbol{\Lambda}_{\bar{\mc S}}$ designate the block-diagonal matrix given by
\begin{align}~\label{equation-definition-T}
\boldsymbol{\Lambda}_{\bar{\mc S}} :=
\begin{bmatrix}
\dv 0 & \dv 0 \\
\dv 0 & \mathrm{diag}(\{ \dv\Sigma_k - \dv\Sigma_k \dv\Omega_k \dv\Sigma_k \}_{k\in\mc S^c})
\end{bmatrix} 
\end{align}
where $\dv 0$ in the principal diagonal elements is the $n_0{\times}n_0$-all zero matrix.

\noindent The following theorem gives an explicit characterization of $\mc R_{\text{VG-HT}}$.  

\begin{theorem}~\label{theorem-rate-exponent-region-gaussian-hypothesis-testing-against-conditional-independence}
The rate-exponent region $\mc R_{\text{VG-HT}}$ of the vector Gaussian hypothesis testing against conditional independence problem is given by the set of all non-negative tuples $(R_1,\ldots, R_K,E)$ that satisfy, for all subsets $\mc S \subseteq \mc K$,  
\begin{align*}
E &\leq \sum_{k \in \mc S} \big(R_k + \log \left| \dv I - \dv\Omega_k \dv\Sigma_k \right| \big) - \log \left|\dv I + \dv\Sigma_{\dv x}\dv H_0^{\dagger}\dv\Sigma_0^{-1}\dv H_0 \right| \nonumber\\
&\vspace{0.2cm} + \log \left| \dv I + \dv\Sigma_{\dv x} \dv H_{\bar{\mc S}}^\dagger \dv\Sigma_{\dv n_{\bar{\mc S}}}^{-1}\big( \dv I - \boldsymbol{\Lambda}_{\bar{\mc S}} \dv\Sigma_{\dv n_{\bar{\mc S}}}^{-1}\big)\dv H_{\bar{\mc S}}\right| 
\end{align*}
for matrices $\{\dv\Omega_k\}_{k=1}^K$ such that $\dv 0 \preceq \dv\Omega_k \preceq \dv\Sigma_k^{-1}$, where $\bar{\mc S}=\{0\} \cup \mc S^c$ and $\boldsymbol{\Lambda}_{\bar{\mc S}}$ is given by~\eqref{equation-definition-T}. \qedblack
\end{theorem}

\begin{proof}
The proof of Theorem~\ref{theorem-rate-exponent-region-gaussian-hypothesis-testing-against-conditional-independence} is given in Section~\ref{secV_subsecB}. 
\end{proof}

In what follows, we elaborate on two special cases of Theorem~\ref{theorem-rate-exponent-region-gaussian-hypothesis-testing-against-conditional-independence}, i) the one-encoder vector Gaussian testing against conditional independence problem (i.e., $K=1$) and ii) the $K$-encoder scalar Gaussian testing against independence problem.

i) Let us first consider the case $K=1$. In this case, the Markov chain~\eqref{markov-chain-assumption-gaussian-hypothesis-testing-model} which is to be satisfied under the null hypothesis is non-restrictive; and Theorem~\ref{theorem-rate-exponent-region-gaussian-hypothesis-testing-against-conditional-independence} then provides a complete solution of the (general) one-encoder vector Gaussian testing against conditional independence problem. More precisely, in this case the optimal trade-off between rate and Type II error exponent is given by the set of pairs $(R_1,E)$ that satisfy
\begin{subequations}
\begin{align}
E & \leq R_1 + \log \left|\dv I - \dv\Omega_1 \dv\Sigma_1 \right| \\
E &\leq \log \left| \dv I + \dv\Sigma_{\dv x} \dv H_{\{0,1\}}^\dagger \dv\Sigma_{\dv n_{\{0,1\}}}^{-1}\big( \dv I - \boldsymbol{\Lambda}_{\{0,1\}} \dv\Sigma_{\dv n_{\{0,1\}}}^{-1}\big)\dv H_{\{0,1\}} \right| \nonumber\\
& \qquad - \log \left|\dv I + \dv\Sigma_{\dv x}\dv H_0^{\dagger}\dv\Sigma_0^{-1}\dv H_0 \right|,
\end{align}
\label{rate-exponent-region-one-encoder-testing-against-conditional-independence-vector-Gaussian-case}
\end{subequations}
for some $n_1{\times}n_1$ matrix $\dv\Omega_1$ such that $\dv 0 \preceq \dv\Omega_1 \preceq \dv\Sigma_1^{-1}$, where $\dv H_{\{0,1\}}=[\dv H^{\dag}_0, \dv H^{\dag}_1]^{\dag}$, $\dv\Sigma_{\dv n_{\{0,1\}}}$ is the covariance matrix of noise $(\dv N_0, \dv N_1)$ and 
\begin{align}
\boldsymbol{\Lambda}_{\{0,1\}} :=
\begin{bmatrix}
\dv 0 & \dv 0 \\
\dv 0 & \dv\Sigma_1 - \dv\Sigma_1 \dv\Omega_1 \dv\Sigma_1
\end{bmatrix}
\end{align}
with the $\dv 0$ in its principal diagonal denoting the $n_0{\times}n_0$-all zero matrix. In particular, for the setting of testing against independence, i.e., $\dv Y_0=\emptyset$ and the decoder's task reduced to guessing whether $\dv Y_1$ and $\dv X$ are independent or not, the optimal trade-off expressed by~\eqref{rate-exponent-region-one-encoder-testing-against-conditional-independence-vector-Gaussian-case} reduces to the set of $(R_1,E)$ pairs that satisfy, for some $n_1{\times}n_1$ matrix $\dv\Omega_1$ such that $\dv 0 \preceq \dv\Omega_1 \preceq \dv\Sigma_1^{-1}$,
\begin{equation}
E \leq \min \left\{R_1 + \log \left|\dv I - \dv\Omega_1 \dv\Sigma_1 \right|, \:\: \log \left| \dv I + \dv\Sigma_{\dv x}  \dv H_1^\dagger \dv\Omega_1 \dv H_1 \right| \right\}.
\label{optimal-rate-exponent-region-Gaussian-one-encoder-testing-against-independence}
\end{equation}
Observe that~\eqref{rate-exponent-region-one-encoder-testing-against-conditional-independence-vector-Gaussian-case} is the counter-part, to the vector Gaussian setting, of the result of~\cite[Theorem 3]{RW12} which provides a single-letter formula for the Type II error exponent for the one-encoder DM testing against conditional independence problem. Similarly,~\eqref{optimal-rate-exponent-region-Gaussian-one-encoder-testing-against-independence} is the solution of the vector Gaussian version of the one-encoder DM testing against independence problem which is studied, and solved, by Ahlswede and Csiszar in~\cite[Theorem 2]{AC86}. Also, we mention that, perhaps non-intuitive, in the one-encoder vector Gaussian testing against independence problem swapping the roles of $\dv Y_1$ and $\dv X$ (i.e., giving $\dv X$ to the encoder and the noisy (under the null hypothesis) $\dv Y_1$ to the decoder) does not result in an increase of the Type II error exponent which is then identical to~\eqref{optimal-rate-exponent-region-Gaussian-one-encoder-testing-against-independence}. Note that this is in sharp contrast with the related\footnote{The connection, which is sometimes misleading, consists in viewing the decoder in the hypothesis testing against independence problem considered here as one that computes a binary-valued function of $(\dv X,\dv Y_1)$.} setting of standard lossy source reproduction, i.e., the decoder aiming to reproduce the source observed at the encoder to within some average squared error distortion level using the sent compression message and its own side information, for which it is easy to see that, for given $R_1$ bits per sample, smaller distortion levels are allowed by having the encoder observe $\dv X$ and the decoder observe $\dv Y_1$, instead of the encoder observing the noisy $\dv Y_1=\dv H_1 \dv X + \dv N_1$ and the decoder observing $\dv X$. 

\vspace{0.2cm}

ii) Consider now the special case of the setup of Theorem~\ref{theorem-rate-exponent-region-gaussian-hypothesis-testing-against-conditional-independence} in which $K \geq 2$, $Y_0 = \emptyset$, and the sources and noises are all scalar complex-valued, i.e., $n_x=1$ and $n_k=1$ for all $k \in \mc K$. The vector $(Y_1,\hdots,Y_K)$ and $X$ are correlated under the null hypothesis $H_0$ and independent under the alternate hypothesis $H_1$, with
 \begin{subequations}
\begin{align}
\label{distributions-under-null-hypothesis-ht-against-independence-scalar-Gaussian}
H_0 \: &: Y_k = X + N_k, \quad\text{for all}\:\: k \in \mc K\\
H_1 \: &: (Y_1,\hdots,Y_K) \:\: \text{independent from}\:\: X.
\label{distributions-under-alternate-hypothesis-ht-against-independence-scalar-Gaussian}
\end{align}
\label{distributions-under-null-and-alternate-hypotheses-ht-against-independence-scalar-Gaussian}
\end{subequations}

\noindent The noises $N_1,\hdots,N_K$ are zero-mean jointly Gaussian, mutually independent and independent from $X$. Also, we assume that the variances $\sigma^2_k$ of noise $N_k$, $k \in \mc K$, and $\sigma^2_X$ of $X$ are all positive. In this case, it can be easily shown that Theorem~\ref{theorem-rate-exponent-region-gaussian-hypothesis-testing-against-conditional-independence} reduces to
\begin{align}
 & \mc R_{\text{SG-HT}} = \Big\{  (R_1,\hdots,R_K,E)\::\: \exists \: (\gamma_1,\hdots,\gamma_K) \in \mathbb{R}^K_{+} \:\: \text{such that} \nonumber\\
&\gamma_k \leq \frac{1}{\sigma^2_k},\:\forall k\in \mc K,\:\: \text{and} \:\: \forall \: \mc S \subseteq \mc K \nonumber\\
&\sum_{k \in \mc S} R_k \geq E + \log\Big[\Big(\Big(1+\sigma^2_X\sum_{k \in \mc S^c}\gamma_k\Big)\prod_{k\in \mc S}(1-\gamma_k \sigma^2_k)\Big)^{-1}\Big]\Big\}.
\label{rate-exponent-region-ht-against-independence-scalar-Gaussian}
\end{align}

\noindent The region $\mc R_{\text{SG-HT}}$ as given by~\eqref{rate-exponent-region-ht-against-independence-scalar-Gaussian} can be used to, e.g., characterize the centralized rate region, i.e., the set of rate vectors $(R_1,\hdots,R_K)$ that achieve the centralized Type II error exponent 
\begin{equation}
I(Y_1,\hdots,Y_K;X) = \sum_{k=1}^K \log \frac{\sigma^2_X}{\sigma^2_k}.
\end{equation}

 We close this section by mentioning that, implicit in Theorem~\ref{theorem-rate-exponent-region-gaussian-hypothesis-testing-against-conditional-independence}, the Quantize-Bin-Test scheme of~\cite[Theorem 1]{RW12} with Gaussian test channels and time-sharing is optimal for the vector Gaussian $K$-encoder hypothesis testing against conditional independence problem~\eqref{distributions-under-null-and-alternate-hypotheses-ht-against-conditional-independence-vector-Gaussian}. Furthermore, we note that Rahman and Wagner also characterized the optimal rate-exponent region of a different\footnote{This problem is related to the Gaussian many-help-one problem~\cite{O05,PTR04,WTV08}. Here, different from the setup of Figure~\ref{fig-distributed-hypothesis-testing}, the source $X$ is observed directly by a \textit{main encoder} who communicates with a detector that observes $Y$ in the aim of making a decision on whether $X$ and $Y$ are independent or not. Also, there are helpers that observe independent noisy versions of $X$ and communicate with the detector in the aim of facilitating that test.} Gaussian hypothesis testing against independence problem, called the Gaussian many-help-one hypothesis testing against independence problem therein, in the case of scalar valued sources~\cite[Theorem 7]{RW12}. Specialized to the case $K=1$, the result of Theorem~\ref{theorem-rate-exponent-region-gaussian-hypothesis-testing-against-conditional-independence} recovers that of~\cite[Theorem 7]{RW12} in the case of no helpers; and extends it to vector-valued sources and testing against conditional independence in that case.

\vspace{-0.4cm}

\section{Proofs}~\label{secV}

\vspace{-0.4cm}

\subsection{Proof of Theorem~\ref{theorem-rate-exponent-region-hypothesis-testing-DM-case}}~\label{secV_subsecA}

\subsubsection{Convese part}

Let a non-negative tuple $(R_1,\hdots,R_K,E) \in \mc R_{\text{HT}}$ be given. Since $\mc R_{\text{HT}} = \xbar{\mc R^{\star}}$, then there must exist a series of non-negative tuples $\{(R^{(m)}_1,\hdots,R^{(m)}_K,E^{(m)})\}_{m \in \mathbb{N}}$ such that
\begin{subequations}
\begin{align}
& (R^{(m)}_1,\hdots,R^{(m)}_K,E^{(m)}) \in \mc R^{\star} \:\:\: \text{for all}\:\: m \in \mathbb{N}, \quad \text{and} \\
& \lim_{m \to \infty} (R^{(m)}_1,\hdots,R^{(m)}_K,E^{(m)}) =  (R_1,\hdots,R_K,E).
\end{align}
\label{equivalence-HT-CEO-proof-direct-part-step1}
\end{subequations}
Fix $\delta' > 0$. Then, $\exists \:\: m_0 \in \mathbb{N}$ such that for all $m \geq m_0$, we have
\begin{subequations}
 \begin{align}
R_k &\geq R^{(m)}_k - \delta'  \:\:\: \text{for all}\:\: k \in \mc K, \quad \text{and} \\
E &\leq E^{(m)} + \delta'.
\end{align}
\label{equivalence-HT-CEO-proof-direct-part-step2}
\end{subequations}
\noindent For $m \geq m_0$, there exist a series $\{n_m\}_{m \in \mathbb{N}}$ and functions  $\{\breve{\phi}^{(n_m)}_k\}_{k \in \mc K}$ such that 
\begin{subequations}
 \begin{align}
 R^{(m)}_k  &\geq \frac{1}{n_m} \log|\breve{\phi}^{(n_m)}_k|   \:\: \text{for all}\:\: k \in \mc K, \:\: \text{and} \\
E^{(m)} &\leq \frac{1}{n_m} I(\{\breve{\phi}^{(n_m)}_k(Y^{n_m}_k)\}_{k \in \mc K};X^{n_m}|Y^{n_m}_0).
\end{align}
\label{equivalence-HT-CEO-proof-direct-part-step3}
\end{subequations}
\noindent Combining~\eqref{equivalence-HT-CEO-proof-direct-part-step2} and~\eqref{equivalence-HT-CEO-proof-direct-part-step3} we get that for all $m \geq m_0$,
\begin{subequations}
 \begin{align}
 R_k  &\geq \frac{1}{n_m} \log|\breve{\phi}^{(n_m)}_k(Y^{n_m}_k)| - \delta'   \:\: \text{for all}\:\: k \in \mc K, \:\: \text{and} \\
E  &\leq \frac{1}{n_m} I(\{\breve{\phi}^{(n_m)}_k(Y^{n_m}_k)\}_{k \in \mc K};X^{n_m}|Y^{n_m}_0) + \delta'.
\end{align}
\label{equivalence-HT-CEO-proof-direct-part-step4}
\end{subequations}
\noindent The second inequality of~\eqref{equivalence-HT-CEO-proof-direct-part-step4} implies that
\begin{equation}
H(X^{n_m} | \{\breve{\phi}^{(n_m)}_k(Y^{n_m}_k)\}_{k \in \mc K}, Y^{n_m}_0) \leq n_m (H(X|Y_0)-E) + n_m \delta'.
\label{equivalence-HT-CEO-proof-direct-part-step5}
\end{equation}

\noindent Let $\mc S \subseteq \mc K$ a given subset of $\mc K$ and $J_k := \breve{\phi}_k^{(n_m)}(Y_k^{n_m})$. Also, define, for $i=1,\ldots,n_m$, the following auxiliary random variables
\begin{equation}	
	U_{k,i} := (J_k, Y_k^{i-1}), \quad Q_i := (X^{i-1}, X_{i+1}^{n_m}, Y_0^{i-1}, Y_{0,i+1}^{n_m}).
	\label{proof-converse-definition-of-auxiliary-random-variables}
\end{equation}
Note that, for all $k \in \mc K$, it holds that $U_{k,i} \mkv Y_{k,i} \mkv (X_i, Y_{0,i}) \mkv Y_{\mc K \setminus k,i} \mkv U_{\mc K \setminus k,i}$ is a Markov chain in this order.

\noindent We have 
\begin{align}
n_m & \sum_{k\in \mc S} R_k \geq \sum_{k \in \mc S} H(J_k) \nonumber\\
& \geq H(J_\mc S) \nonumber\\
& \geq  H(J_\mc S|J_{\mc S^c},Y_0^{n_m}) \nonumber\\
& \geq I(J_{\mc S}; X^{n_m}, Y_\mc S^{n_m}|J_{\mc S^c},Y_0^{n_m}) \nonumber\\   
&= I(J_{\mc S}; X^{n_m}|J_{\mc S^c},Y_0^{n_m}) + I(J_{\mc S}; Y_\mc S^n|X^{n_m},J_{\mc S^c},Y_0^{n_m}) \nonumber\\ 
&= H(X^{n_m}|J_{\mc S^c},Y_0^{n_m}) - H(X^{n_m}|J_{\mc K},Y_0^{n_m}) \nonumber\\
& \qquad + I(J_{\mc S}; Y_\mc S^{n_m}|X^{n_m},J_{\mc S^c},Y_0^{n_m}) \nonumber\\  
&\stackrel{(a)}{\geq} H(X^{n_m}|J_{\mc S^c}, Y_0^{n_m}) - H(X^{n_m}|Y^{n_m}_0)  \nonumber\\
& \qquad + I(J_{\mc S}; Y_\mc S^{n_m}|X^{n_m},J_{\mc S^c},Y_0^{n_m}) + n_m E - n_m \delta' \nonumber\\  
&= \sum_{i=1}^{n_m} H(X_i|J_{\mc S^c},X^{i-1},Y_0^{n_m}) - H(X^{n_m}|Y^{n_m}_0) \nonumber\\
& \qquad + I(J_{\mc S}; Y_\mc S^{n_m}|X^{n_m},J_{\mc S^c},Y_0^{n_m}) + n_m E - n_m \delta' \nonumber\\  
&\stackrel{(b)}{\geq} \sum_{i=1}^{n_m} H(X_i|J_{\mc S^c},X^{i-1},X_{i+1}^{n_m},Y_{\mc S^c}^{i-1},Y_0^{n_m}) - H(X^{n_m}|Y^{n_m}_0) \nonumber\\
& \qquad + I(J_{\mc S}; Y_\mc S^{n_m}|X^{n_m},J_{\mc S^c},Y_0^{n_m}) + n_m E - n_m \delta' \nonumber\\
&\stackrel{(c)}{=} \sum_{i=1}^{n_m} H(X_i|U_{\mc S^c,i},Y_{0,i},Q_i) - H(X^{n_m}|Y^{n_m}_0) \nonumber\\
& \qquad + I(J_{\mc S}; Y_\mc S^{n_m}|X^{n_m},J_{\mc S^c},Y_0^{n_m}) + n_m E - n_m \delta' \nonumber\\
&\stackrel{(d)}{=} I(J_{\mc S}; Y_\mc S^{n_m}|X^{n_m},J_{\mc S^c},Y_0^{n_m}) - \sum_{i=1}^{n_m} I(U_{\mc S^c,i}, X_i|Y_{0,i},Q_i)  \nonumber\\
&\qquad + n_m E - n_m \delta' 
\label{lower-bounding-sum-rate-step1}
\end{align}
where $(a)$ follows by using~\eqref{equivalence-HT-CEO-proof-direct-part-step5}; $(b)$ holds since conditioning reduces entropy; and $(c)$ follows by substituting using~\eqref{proof-converse-definition-of-auxiliary-random-variables}; and $(d)$ holds since $(X^{n_m},Y^{n_m}_0)$ is memoryless and $Q_i$ is independent of $(X_i,Y_{0,i})$ for all $i=1,\hdots,n_m$. 

\noindent The term $I(J_{\mc S}; Y_\mc S^{n_m}|X^{n_m},J_{\mc S^c},Y_0^{n_m})$ on the RHS of~\eqref{lower-bounding-sum-rate-step1} can be lower bounded as
\begin{align}
 I(J_{\mc S}; & Y_{\mc S}^{n_m}|X^{n_m},J_{\mc S^c},Y_0^{n_m}) \stackrel{(a)}{\geq} \sum_{k \in \mc S} I(J_k;Y_k^{n_m}|X^{n_m},Y_0^{n_m}) \nonumber\\
& = \sum_{k \in \mc S} \sum_{i=1}^{n_m} I(J_k;Y_{k,i}|Y_k^{i-1},X^{n_m},Y_0^{n_m}) \nonumber\\
&\stackrel{(b)}{=} \sum_{k \in \mc S} \sum_{i=1}^{n_m} I(J_k,Y_k^{i-1};Y_{k,i}|X^{n_m},Y_0^{n_m}) \nonumber\\
&\stackrel{(c)}{=} \sum_{k \in \mc S} \sum_{i=1}^{n_m} I(U_{k,i};Y_{k,i}|X_i,Y_{0,i},Q_i)
 \label{lower-bounding-sum-rate-step2}
\end{align}
where $(a)$ follows due to the Markov chain $J_k \mkv Y_k^{n_m} \mkv (X^{n_m}, Y_0^{n_m}) \mkv Y_{\mc S \setminus k}^{n_m} \mkv J_{\mc S \setminus k}$ under the hypothesis $H_0$; $(b)$ follows due to the Markov chain $Y_{k,i} \mkv (X^{n_m},Y_0^{n_m}) \mkv Y_k^{i-1}$ under the hypothesis $H_0$; and $(c)$ follows by substituting using~\eqref{proof-converse-definition-of-auxiliary-random-variables}. 

\noindent Then, combining~\eqref{lower-bounding-sum-rate-step1} and~\eqref{lower-bounding-sum-rate-step2}, we get
\begin{align}
 n_m E & \leq \sum_{i=1}^{n_m} I(U_{\mc S^c,i}, X_i|Y_{0,i},Q_i ) + n_m \sum_{k\in \mc S} R_k  \nonumber\\
& \qquad - \sum_{k\in \mc S} \sum_{i=1}^{n_m} I(U_{k,i};Y_{k,i}|X_i,Y_{0,i},Q_i) + n_m \delta'.
\label{upper-bounding-exponent-final-step} 
\end{align}
\noindent Noticing that $\delta'$ in~\eqref{upper-bounding-exponent-final-step} can be chosen arbitrarily small, a standard time-sharing argument completes the proof of the converse part.

\subsubsection{Direct part}

The achievability follows by applying the Quantize-Bin-Test scheme of Rahman and Wagner~\cite[Appendix B]{RW12}. Applied to our model, the rate-exponent region achieved by this scheme, which we denote as $\mc R_{\text{QBT}}$,  is given by the union of all non-negative rate-exponent tuples $(R_1,\hdots,R_K,E)$ for which
\begin{subequations}
\begin{align}
\sum_{k \in \mc S} R_k &\geq I(U_{\mc S}; Y_{\mc S}|U_{\mc S^c}, Y_0, Q), \quad \forall \mc S \subseteq \mc K, \\
E &\leq I(U_{\mc K};X|Y_0).
\end{align}
\label{rate-exponent-region-quantize-bin-test-scheme}
\end{subequations}

\noindent Through submodularity arguments that are essentially similar to in~\cite[Appendix B]{CW14} (see also~\cite{E-AZCS17a} and~\cite[Appendix IV]{E-AZCS19a}), and which we omit here for brevity, the region $\mc R_{\text{QBT}}$ can be shown to be equivalent to the region $\mc R_{\text{HT}}$ as stated in Theorem~\ref{theorem-rate-exponent-region-hypothesis-testing-DM-case}.

\subsection{Proof of Theorem~\ref{theorem-rate-exponent-region-gaussian-hypothesis-testing-against-conditional-independence}}~\label{secV_subsecB}

\vspace{-0.2cm}

Let an achievable tuple $(R_1,\hdots,R_K,E)$ for the memroryless vector Gaussian hypothesis testing against conditional independence problem of Section~\ref{secIV} be given. By a standard extension of the result of Theorem~\ref{theorem-rate-exponent-region-hypothesis-testing-DM-case} to the continuous alphabet case (through standard discretezation arguments), there must exist a.r.v. $(U_1,\hdots,U_K,Q)$ with distribution that factorizes as 
 \begin{align}
& P_{\dv X, \dv Y_0, \dv Y_{\mc K}, U_{\mc K}, Q}(\dv x, \dv y_0, \dv y_{\mc K},u_{\mc K},q) = P_Q(q)  P_{\dv X, \dv Y_0}(\dv x, \dv y_0) \nonumber\\
& \hspace{0.1cm} = \prod_{k=1}^K P_{\dv Y_k|\dv X, \dv Y_0}(\dv y_k|\dv x,\dv y_0) \: \prod_{k=1}^{K} P_{U_k|\dv Y_k,Q}(u_k|\dv y_k,q).
\end{align}
such that for all $\mc S \subseteq \mc K$, 
\begin{equation}
E - \sum_{k \in \mc S} R_k \leq I(U_{\mc S^c};\dv X|\dv Y_0,Q) - \sum_{k \in \mc S} I(\dv Y_k;U_k|\dv X,\dv Y_0,Q).
\label{equivalent-representation-rate-exponent-region-DM-case}
\end{equation}
The converse proof of Theorem~\ref{theorem-rate-exponent-region-gaussian-hypothesis-testing-against-conditional-independence} relies on deriving an upper bound on the RHS of~\eqref{equivalent-representation-rate-exponent-region-DM-case}.  In doing so, we use the technique of~\cite[Theorem 8]{EU14} which relies on the de Bruijn identity and the properties of Fisher information; and extend the argument to account for the time-sharing variable $Q$ and side information $\dv Y_0$.

For convenience, we first state the following lemma.

\begin{lemma}{\cite{DCT91,EU14}}~\label{lemma-fisher}
Let $(\mathbf{X,Y})$  be a pair of random vectors with pmf $p(\mathbf{x},\mathbf{y})$. We have
\begin{equation*}
\log|(\pi e) \dv J^{-1}(\dv X|\dv Y)| \leq h(\dv X|\dv Y) \leq \log|(\pi e) \mathrm{mmse}(\dv X|\dv Y)|
\end{equation*}
where the conditional Fisher information matrix is defined as
\begin{equation*}
\dv J(\dv X|\dv Y) := \mathbb{E} [\nabla\log p(\dv X|\dv Y) \nabla\log p(\dv X|\dv Y)^\dagger]
\end{equation*}
and the minimum mean squared error (MMSE) matrix is 
\begin{equation*}
\mathrm{mmse}(\dv X|\dv Y) := \mathbb{E} [(\dv X-\mathbb{E}[\dv X|\dv Y])(\dv X-\mathbb{E} [\dv X|\dv Y])^\dagger]. \bqed 
\end{equation*}
\end{lemma}

Fix $q\in \mc{Q}$, $\mc S \subseteq \mc Q$. Also, let $\dv 0 \preceq \dv\Omega_{k,q} \preceq \dv\Sigma_k^{-1}$ and  
\begin{equation}~\label{equation-outer-1-mmse}
\mathrm{mmse}(\dv Y_k|\dv X, U_{k,q}, \dv Y_0, q) = \dv\Sigma_k - \dv\Sigma_k \dv\Omega_{k,q} \dv\Sigma_k.
\end{equation}
Such $\dv\Omega_{k,q}$ always exists since 
\begin{equation*}~\label{equation-outer-2-covariance-noise}
\dv 0 \preceq \mathrm{mmse}(\dv Y_k|\dv X,U_{k,q},\dv Y_0,q) \preceq \dv\Sigma_{\dv y_k|(\dv x, \dv y_0)} = \dv\Sigma_k.
\end{equation*} 

\noindent Then, we have
\begin{align}
& I(\dv Y_k; U_k|\dv X, \dv Y_0, Q=q) \nonumber\\
&\qquad =  \log|(\pi e)\dv\Sigma_k| - h(\dv Y_k|\dv X, U_{k,q}, \dv Y_0, Q=q) \nonumber\\
&\qquad \stackrel{(a)}{\geq} \log|\dv\Sigma_k| - \log|\mathrm{mmse}(\dv Y_k|\dv X, U_{k,q}, \dv Y_0, Q=q)| \nonumber\\
&\qquad \stackrel{(b)}{=} -\log|\dv I- \dv\Omega_{k,q}\dv\Sigma_k| \label{equation-Gausss-CEO-first-inequality-q}
\end{align}
where $(a)$ is due to Lemma~\ref{lemma-fisher}; and $(b)$ is due to~\eqref{equation-outer-1-mmse}.

\noindent Now, let the matrix $\dv\Lambda_{\bar{\mc S},q}$ be defined as  
\begin{align}~\label{equation-definition-Tq}
\dv\Lambda_{\bar{\mc S},q} :=
\begin{bmatrix}
\dv 0 & \dv 0 \\
\dv 0 & \mathrm{diag}(\{ \dv\Sigma_k - \dv\Sigma_k \dv\Omega_{k,q} \dv\Sigma_k \}_{k\in\mc S^c})
\end{bmatrix}. 
\end{align}
Then, we have 
\begin{align}
& I(U_{\mc S^c}; \dv X|\dv Y_0,Q=q) = h(\dv X|\dv Y_0) - h(\dv X|U_{S^c,q}, \dv Y_0, Q=q) \nonumber\\
 &\: \stackrel{(a)}{\leq} h(\dv X|\dv Y_0) - \log|(\pi e) \dv J^{-1}(\dv X| \dv U_{S^c,q}, \dv Y_0, q)| \nonumber\\
&\: \stackrel{(b)}{=} h(\dv X|\dv Y_0) \nonumber\\
&\qquad - \log\left| (\pi e) \left( \dv\Sigma_{\dv x}^{-1} + \dv H_{\bar{\mc S}}^\dagger \dv\Sigma_{\dv n_{\bar{\mc S}}}^{-1} \big( \dv I - \dv\Lambda_{\bar{\mc S},q} \dv\Sigma_{\dv n_{\bar{\mc S}}}^{-1} \big) \dv H_{\bar{\mc S}} \right)^{-1} \right| \label{equation-Gausss-CEO-second-inequality-q}
\end{align}
where $(a)$ follows by using Lemma~\ref{lemma-fisher}; and for $(b)$ holds by using the equality 
\begin{equation}~\label{equation-Fisher-equality}
\dv J(\dv X|U_{S^c,q}, \dv Y_0, q) = \dv\Sigma_{\dv x}^{-1} + \dv H_{\bar{\mc S}}^\dagger \dv\Sigma_{\dv n_{\bar{\mc S}}}^{-1} \big( \dv I - \dv\Lambda_{\bar{\mc S},q} \dv\Sigma_{\dv n_{\bar{\mc S}}}^{-1} \big) \dv H_{\bar{\mc S}}. 
\end{equation}
the proof of which uses a connection between MMSE and Fisher information as shown next. More precisely, for the proof of~\eqref{equation-Fisher-equality} first recall de Brujin identity which relates Fisher information and MMSE. 

\begin{lemma}{\cite{EU14}}~\label{lemma-Brujin}
Let $(\dv V_1,\dv V_2)$ be a random vector with finite second moments and $\dv Z\sim\mc{CN}(\dv 0, \dv\Sigma_{\dv z})$ independent of $(\dv V_1,\dv V_2)$. Then
\begin{equation*}
\mathrm{mmse}(\dv V_2|\dv V_1,\dv V_2+\dv Z) = \dv\Sigma_{\dv z} - \dv\Sigma_{\dv z} \dv J(\dv V_2+\dv Z|\dv V_1) \dv\Sigma_{\dv z}. \bqed 
\end{equation*}
\end{lemma}

From MMSE estimation of Gaussian random vectors, we have  
\begin{equation}~\label{equation-outer-3}
\dv X = \mathbb{E} [\dv X|\dv Y_{\bar{\mc S}}] + \dv W_{\bar{\mc S}} = \dv G_{\bar{\mc S}} \dv Y_{\bar{\mc S}} + \dv W_{\bar{\mc S}} 
\end{equation}
where $\dv G_{\bar{\mc S}} := \dv\Sigma_{\dv w_{\bar{\mc S}}} \dv H_{\bar{\mc S}}^\dagger \dv\Sigma_{\dv n_{\bar{\mc S}}}^{-1}$, and $\dv W_{\bar{\mc S}} \sim \mathcal{CN}(\dv 0, \dv\Sigma_{\dv w_{\bar{\mc S}}})$ is a Gaussian vector that is independent of $\dv Y_{\bar{\mc S}}$ and  
\begin{equation}~\label{equation-covariance-w}
\dv\Sigma_{\dv w_{\bar{\mc S}}}^{-1} :=  \dv\Sigma_{\dv x}^{-1} + \dv H_{\bar{\mc S}}^\dagger \dv\Sigma_{\dv n_{\bar{\mc S}}}^{-1} \dv H_{\bar{\mc S}}. 
\end{equation}

\noindent Next, we show that the cross-terms of $\mathrm{mmse}\left(\dv Y_{\mc S^c}|\dv X,U_{\mc S^{c},q},\dv Y_0,q \right)$ are zero. For $i\in\mc S^c$ and $j\neq i$, we have 
\begin{align}
& \mathbb{E} \big[(Y_i-\mathbb{E}[Y_i|\dv X,U_{\mc S^{c},q},\dv Y_0,q])(Y_j-\mathbb{E}[Y_j|\dv X,U_{\mc S^{c},q},\dv Y_0,q])^\dagger\big] \nonumber\\
&\stackrel{(a)}{=}  \mathbb{E}\bigg[ \mathbb{E}\big[ (Y_i-\mathbb{E}[Y_i|\dv X,U_{\mc S^{c},q},\dv Y_0,q])\nonumber\\
& \qquad \qquad {\times} (Y_j-\mathbb{E}[Y_j|\dv X,U_{\mc S^{c},q},\dv Y_0,q])^\dagger|\dv X,\dv Y_0 \big] \bigg] \nonumber\\
&\stackrel{(b)}{=}  \mathbb{E}\bigg[ \mathbb{E}\big[ (Y_i-\mathbb{E}[Y_i|\dv X,U_{\mc S^{c},q},\dv Y_0,q])|\dv X,\dv Y_0 \big] \nonumber\\
&\qquad \qquad {\times}\mathbb{E}\big[(Y_j-\mathbb{E}[Y_j|\dv X,U_{\mc S^{c},q},\dv Y_0,q])^\dagger|\dv X,\dv Y_0 \big] \bigg] \nonumber\\
& = \dv 0,
\label{equation-cross-terms}
\end{align}
where $(a)$ is due to the law of total expectation; $(b)$ is due to the Markov chain $\dv Y_k \mkv (\dv X,\dv Y_0) \mkv \dv Y_{\mc K \setminus k}$. Then, we have
\begin{align}
& \mathrm{mmse}\big(\dv G_{\bar{\mc S}} \dv Y_{\bar{\mc S}} \big|\dv X,U_{\mc S^c,q},\dv Y_0,q \big) \nonumber\\
& \qquad = \dv G_{\bar{\mc S}}  \: \mathrm{mmse}\left(\dv Y_{\bar{\mc S}}|\dv X,U_{\mc S^c,q},\dv Y_0,q \right) \dv G_{\bar{\mc S}}^\dagger \nonumber\\
&\qquad \stackrel{(a)}{=} 
\dv G_{\bar{\mc S}}
\begin{bmatrix}
\dv 0 & \dv 0 \\
\dv 0 & \mathrm{diag}(\{\mathrm{mmse}(\dv Y_k|\dv X,U_{\mc S^c,q},\dv Y_0,q)\}_{k\in\mc S^c})
\end{bmatrix}
\dv G_{\bar{\mc S}}^\dagger \nonumber\\ 
&\qquad \stackrel{(b)}{=} \dv G_{\bar{\mc S}} \dv\Lambda_{\bar{\mc S},q} \dv G_{\bar{\mc S}}^\dagger \label{equation-outer-4}
\end{align}
where $(a)$ follows since the cross-terms are zero as shown in~\eqref{equation-cross-terms}; and $(b)$ follows due to~\eqref{equation-outer-1-mmse} and the definition of $\dv\Lambda_{\bar{\mc S},q}$ given in~\eqref{equation-definition-Tq}.

\vspace{0.2cm}

\noindent We note that $\dv W_{\bar{\mc S}}$ is independent of $\dv Y_{\bar{\mc S}}=(\dv Y_0, \dv Y_{{\mc S}^c})$; and, with the Markov chain $U_{{\mc S}^c} \mkv \dv Y_{{\mc S}^c} \mkv (\dv X, \dv Y_0)$, which itself implies $U_{{\mc S}^c} \mkv \dv Y_{{\mc S}^c} \mkv (\dv X, \dv Y_0,\dv W_{\bar{\mc S}})$, this yields that $\dv W_{\bar{\mc S}}$ is independent of $U_{{\mc S}^c}$. Thus, $\dv W_{\bar{\mc S}}$ is independent of $(\dv G_{\bar{\mc S}} \dv Y_{\bar{\mc S}},U_{{\mc S}^c},\dv Y_0,Q)$. Applying Lemma~\ref{lemma-Brujin} with $\dv V_1 := (U_{{\mc S}^c},\dv Y_0,Q)$, $\dv V2 :=\dv G_{\bar{\mc S}} \dv Y_{\bar{\mc S}}$ and $\dv Z := \dv W_{\bar{\mc S}}$, we get   
\begin{align*}    
& \dv J(\dv X|U_{S^c,q},\dv Y_0,q) \nonumber\\
&= \dv\Sigma_{\dv w_{\bar{\mc S}}}^{-1} - \dv\Sigma_{\dv w_{\bar{\mc S}}}^{-1} \: \mathrm{mmse} \big( \dv G_{\bar{\mc S}} \dv Y_{\bar{\mc S}} \big| \dv X,U_{\mc S^{c},q},\dv Y_0,q \big) \dv\Sigma_{\dv w_{\bar{\mc S}}}^{-1} \\
&\qquad \stackrel{(a)}{=} \dv\Sigma_{\dv w_{\bar{\mc S}}}^{-1} - \dv\Sigma_{\dv w_{\bar{\mc S}}}^{-1} \dv G_{\bar{\mc S}} \dv\Lambda_{\bar{\mc S},q} \dv G_{\bar{\mc S}}^\dagger \dv\Sigma_{\dv w_{\bar{\mc S}}}^{-1} \\
&\qquad \stackrel{(b)}{=} \dv\Sigma_{\dv x}^{-1} + \dv H_{\bar{\mc S}}^\dagger \dv\Sigma_{\dv n_{\bar{\mc S}}}^{-1} \dv H_{\bar{\mc S}} -  \dv H_{\bar{\mc S}}^\dagger \dv\Sigma_{\dv n_{\bar{\mc S}}}^{-1} \dv\Lambda_{\bar{\mc S},q} \dv\Sigma_{\dv n_{\bar{\mc S}}}^{-1} \dv H_{\bar{\mc S}}\\
&\qquad = \dv\Sigma_{\dv x}^{-1} + \dv H_{\bar{\mc S}}^\dagger \dv\Sigma_{\dv n_{\bar{\mc S}}}^{-1} \big( \dv I - \dv\Lambda_{\bar{\mc S},q} \dv\Sigma_{\dv n_{\bar{\mc S}}}^{-1}  \big) \dv H_{\bar{\mc S}\bar{\mc S}},
\end{align*}
where $(a)$ is due to~\eqref{equation-outer-4}; and $(b)$ follows due to the definitions of $\dv\Sigma_{\dv w_{\bar{\mc S}}}^{-1}$ and $\dv G_{\bar{\mc S}}$.

\noindent Next, we average the expression in ~\eqref{equation-Gausss-CEO-first-inequality-q} and~\eqref{equation-Gausss-CEO-second-inequality-q} over the time-sharing $Q$ and letting $\dv\Omega_k := \sum_{q\in \mc Q}p(q) \dv\Omega_{k,q}$, we obtain the lower bound
\begin{align}
I(\dv Y_k;\dv U_k|\dv X,\dv Y_0,Q) &= \sum_{q \in \mc Q} p(q) I(\dv Y_k;\dv U_k|\dv X,\dv Y_0,Q=q)  \nonumber\\
&\stackrel{(a)}{\geq} - \sum_{q \in \mc Q} p(q) \log|\dv I- \dv\Omega_{k,q} \dv\Sigma_k| \nonumber\\
&\stackrel{(b)}{\geq} -\log |\dv I - \sum_{q \in \mc Q} p(q) \dv\Omega_{k,q} \dv\Sigma_k| \nonumber\\  
&= -\log |\dv I- \dv\Omega_k \dv\Sigma_k| \label{equation-Gausss-CEO-first-inequality}
\end{align}
where $(a)$ follows from~\eqref{equation-Gausss-CEO-first-inequality-q}; and $(b)$ follows from the concavity of the log-det function and Jensen's Inequality. 

\noindent Besides,  we have
\begin{align}
& I(U_{\mc S^c}; \dv X|\dv Y_0,Q=q) = h(\dv X|\dv Y_0) - \sum_{q \in \mc Q} p(q) h(\dv X|U_{S^c,q}, \dv Y_0, Q=q) \nonumber\\ 
&\stackrel{(a)}{\leq} h(\dv X|\dv Y_0) \nonumber\\
& \qquad - \sum_{q \in \mc Q} p(q) \log\left| (\pi e) \left(\dv\Sigma_{\dv x}^{-1} + \dv H_{\bar{\mc S}}^\dagger \dv\Sigma_{\dv n_{\bar{\mc S}}}^{-1} \big( \dv I - \dv\Lambda_{\bar{\mc S},q} \dv\Sigma_{\dv n_{\bar{\mc S}}}^{-1}  \big) \dv H_{\bar{\mc S}} \right)^{-1} \right| \nonumber\\
&\stackrel{(b)}{\leq} h(\dv X|\dv Y_0) - \log\left| (\pi e) \left( \dv\Sigma_{\dv x}^{-1} + \dv H_{\bar{\mc S}}^\dagger \dv\Sigma_{\dv n_{\bar{\mc S}}}^{-1} \big( \dv I - \dv\Lambda_{\bar{\mc S}} \dv\Sigma_{\dv n_{\bar{\mc S}}}^{-1}  \big) \dv H_{\bar{\mc S}} \right)^{-1} \right|, \label{equation-Gausss-CEO-second-inequality}  
\end{align}
where $(a)$ is due to~\eqref{equation-Gausss-CEO-second-inequality-q}; and $(b)$ is due to the concavity of the log-det function and Jensen's inequality and the definition of $\dv\Lambda_{\bar{\mc S}}$ given in~\eqref{equation-definition-T}. 

Finally, combining~\eqref{equation-Gausss-CEO-first-inequality} and~\eqref{equation-Gausss-CEO-second-inequality}, noting that $\dv\Omega_k = \sum_{q\in \mathcal{Q}}p(q) \dv\Omega_{k,q} \preceq\mathbf{\Sigma}_{k}^{-1}$ since $\mathbf{0} \preceq \dv\Omega_{k,q} \preceq\mathbf{\Sigma}_{k}^{-1}$, and taking the union over $\dv\Omega_k$ satisfying $\mathbf{0} \preceq \dv\Omega_k \preceq\mathbf{\Sigma}_{k}^{-1}$.

\bibliographystyle{IEEEtran}
\bibliography{draft-paper-ht-arxiv-version}
\end{document}